\documentclass[12pt,journal,onecolumn]{IEEEtran}
\usepackage{amssymb}
\usepackage{amsmath}
\usepackage{longtable}
\newtheorem{theorem}{Theorem}[section]
\newtheorem{lemma}[theorem]{Lemma}

\newtheorem{example}{Example}

\newcommand\myatop[2]{\genfrac{}{}{0pt}{}{#1}{#2}}

\usepackage[para]{threeparttable}

\newcommand{\tr}{{\mathrm{Tr}}}

\newcommand{\gf}{{\mathrm{GF}}}
\newcommand{\PG}{{\mathrm{PG}}}

\newcommand{\bC}{{\mathbb{C}}}

\newcommand{\C}{{\mathcal{C}}}
\newcommand{\cA}{{\mathcal{A}}}

\newcommand{\bc}{{\mathbf{c}}}

\newcommand{\bg}{{\mathbf{g}}}

\newcommand{\bx}{{\mathbf{x}}}

\newenvironment{proof}[1][Proof]{\begin{trivlist}
\item[\hskip \labelsep {\bfseries #1}]}{\end{trivlist}}

\makeatletter

\newcommand{\Rmnum}[1]{\expandafter\@slowromancap\romannumeral #1@}
\makeatother

\ifCLASSINFOpdf

\else

\fi

\begin{document}
%
\title{Optimal Binary Linear Codes from Maximal Arcs}

\author{Ziling Heng\thanks{Z. Heng is with the School of Science, Chang'an University, Xi'an 710064, China (email: zilingheng@163.com)},
Cunsheng Ding\thanks{C. Ding is with the Department of Computer Science
                           and Engineering, The Hong Kong University of Science and Technology,
Clear Water Bay, Kowloon, Hong Kong, China (email: cding@ust.hk)},
and
Weiqiong Wang\thanks{W. Wang is with the School of Science, Chang'an University, Xi'an 710064, China  (e-mail:
wqwang@chd.edu.cn).}
}

\maketitle

\begin{abstract}
The binary Hamming codes with parameters $[2^m-1, 2^m-1-m, 3]$ are perfect.
Their extended codes have parameters $[2^m, 2^m-1-m, 4]$ and are distance-optimal.
The first objective of this paper is to construct a class of binary linear codes
with parameters $[2^{m+s}+2^s-2^m,2^{m+s}+2^s-2^m-2m-2,4]$, which have better information rates than
the class of extended binary Hamming codes, and are also distance-optimal.
The second objective is to construct a class of distance-optimal binary codes
with parameters $[2^m+2, 2^m-2m, 6]$. Both classes of binary linear codes have
new parameters.
\end{abstract}


%
\IEEEpeerreviewmaketitle

\section{Introduction}

Let $q$ be a prime power and $\gf(q)$ the finite field with $q$ elements. Let $n, k, d$ be positive integers.
An $[n,\, k,\, d]$ \emph{code} $\C$ over $\gf(q)$ is a $k$-dimensional subspace of $\gf(q)^n$ with minimum (Hamming) distance $d$. The information rate of $\C$ is defined as $k/n$.
Let $A_i$ denote the number of codewords with Hamming weight $i$ in a code
$\C$ of length $n$. The {\em weight enumerator} of $\C$ is defined by
$
1+A_1z+A_2z^2+ \cdots + A_nz^n.
$
The sequence $(1, A_1, A_2, \cdots, A_n)$ is called the \emph{weight distribution} of
the code $\C$.
A code $\C$ is said to be a $t$-weight code  if the number of nonzero
$A_i$ in the sequence $(A_1, A_2, \cdots, A_n)$ is equal to $t$.

The \emph{dual code} of an $[n,\, k,\, d]$ \emph{code} $\C$ over $\gf(q)$, denoted by $\C^\perp$,
is defined by
$$
\C^\perp:=\{\bx \in \gf(q)^n: \bx \cdot \bc = 0 \ \forall \ \bc \in \C\},
$$
where $\bx \cdot \bc$ denotes the standard inner product of the two vectors.
The dual $\C^\perp$ has dimension $n-k$. The minimum distance of $\C^\perp$ is called the
\emph{dual distance} of $\C$. The \emph{extended code} of an $[n,k,d]$ linear code $C$ is defined by
$$\overline{\C}=\left\{(c_1,c_2,\cdots,c_{n+1}):(c_1,c_2,\cdots,c_n)\in \C\mbox{ with }\sum_{i=1}^{n+1}c_i=0\right\}.$$
Then $\overline{\C}$ is an $[n+1,k,\overline{d}]$ code where $\overline{d}=d$ or $d+1$.

An $[n,k,d]$ code over $\gf(q)$ is said to be \emph{distance-optimal} if no $[n,k,d+1]$ code over $\gf(q)$ exists and \emph{almost distance-optimal} if there exists an $[n,k,d+1]$ distance-optimal code over $\gf(q)$.
An $[n,k,d]$ code over $\gf(q)$ is said to be \emph{dimension-optimal} if no $[n,k+1,d]$ code over $\gf(q)$ exists. A code is \emph{optimal} if the parameters of the code meet a bound on
linear codes. Optimal codes are interesting in both theory and practice. The well known \emph{sphere-packing bound} of a $q$-ary $(n,M,d)$ code with $M$ codewords is given by
$$q^n\geq M\sum_{i=0}^{\lfloor\frac{d-1}{2}\rfloor}(q-1)^i\binom{n}{i},$$ where $\lfloor\cdot\rfloor$ denotes the floor function.
An $(n,M,d)$ code is said to be \emph{prefect} if its parameters achieve the sphere-packing bound.
 The only infinite family of perfect binary linear codes are the binary Hamming codes with parameters $[2^m-1, 2^m-1-m, 3]$. The extended binary Hamming codes have parameters $[2^m, 2^m-1-m, 4]$ and are distance-optimal. The motivation of this paper is to search for a class of binary linear codes which are better than the extended binary Hamming codes. The first objective is to present
a class of binary linear codes with parameters $[2^{m+s}+2^s-2^m,2^{m+s}+2^s-2^m-2m-2,4]$, which have better information rates
than the class of extended binary Hamming codes. The second objective of this paper is to
construct a class of distance-optimal binary linear codes with parameters $[2^m+2, 2^m-2m, 6]$.
To this end, subfield, extension and augmentation techniques are employed.

\section{Preliminaries}\label{sect-pre}

\subsection{Group characters and character sums}
Now we recall characters and some character sums over finite fields which will be needed later.

Let $p$ be a prime and $q=p^m$. Let $\gf(q)$ be the finite field with $q$ elements and $\alpha$ a primitive element of $\gf(q)$. The trace function $\tr_{q/p}$
is the homomorphism from $\gf(q)$ onto $\gf(p)$ defined by
$$\tr_{q/p}(x)=\sum_{i=0}^{m-1}x^{p^{i}},\ x\in \gf(q).$$ Denote by $\zeta_p$ the primitive $p$-th root of complex unity.

An \emph{additive character} of $\gf(q)$ is a function $\chi$ from the additive group $(\gf(q),+)$ to the multiplicative group $\bC^{*}$ such that
$$\chi(x+y)=\chi(x)\chi(y),\ x,y\in \gf(q),$$ where $\bC^{*}$ denotes the set of all nonzero complex numbers. For any $a\in \gf(q)$, the function
$$\chi_{a}(x)=\zeta_{p}^{\tr_{q/p}(ax)},\ x\in \gf(q),$$ defines an additive character of $\gf(q)$. In addition, $\{\chi_{a}:a\in \gf(q)\}$ is a group containing all the additive characters of $\gf(q)$. It is clear that $\chi_0(x)=1$ for all $x\in \gf(q)$ and $\chi_0$ is referred to as the trivial additive character of $\gf(q)$. If $a=1$, we call $\chi_1$ the canonical additive character of $\gf(q)$. Clearly, $\chi_a(x)=\chi_1(ax)$. The orthogonality  relation of additive characters is given by
$$\sum_{x\in \gf(q)}\chi_1(ax)=\left\{
\begin{array}{rl}
q    &   \mbox{ for }a=0,\\
0    &   \mbox{ for }a\in \gf(q)^*.
\end{array} \right. $$

Let $\gf(q)^*=\gf(q)\setminus \{0\}$. A \emph{character} $\psi$ of the multiplicative group $\gf(q)^*$ is a homomorphism from  $\gf(q)^*$  to $\bC^{*}$ satisfying $\psi(xy)=\psi(x)\psi(y)$ for all $(x,y)\in \gf(q)^*\times \gf(q)^*$. The multiplication of two characters $\psi,\psi'$ is defined by $(\psi\psi')(x)=\psi(x)\psi'(x)$ for $x\in \gf(q)^*$. All the characters of $\gf(q)^*$ can be given by
$$\psi_{j}(\alpha^k)=\zeta_{q-1}^{jk}\mbox{ for }k=0,1,\cdots,q-1,$$
where $0\leq j \leq q-2$. Then $\{\psi_j:0\leq j \leq q-2\}$ is a group under the multiplication of characters and its elements are called \emph{multiplicative characters} of $\gf(q)$. In particular, $\psi_0$ is called the trivial multiplicative character of  $\gf(q)$. The orthogonality relation of multiplicative characters is given by
$$\sum_{x\in \gf(q)^*}\psi_j(x)=\left\{
\begin{array}{rl}
q-1    &   \mbox{ for }j=0,\\
0    &   \mbox{ for }j\neq 0.
\end{array} \right. $$

Let $\chi$ be a nontrivial additive character of $\gf(q)$ and $f\in \gf(q)[x]$ a polynomial of positive degree. \emph{Weil sums} are a special class of character sums in the form
$$\sum_{c\in \gf(q)}\chi(f(c)).$$ The problem of evaluating
such character sums explicitly is very difficult in general. However, Weil sums can be treated in some special cases (see \cite[Section 4 in Chapter 5]{LN}).

If $f$ is an affine $p$-polynomial over $\gf(q)$, the Weil sums can be evaluated explicitly.

\begin{lemma}\label{lem-p-polynomial}\cite[Theorem 5.34]{LN}
Let $q=p^m$ and let
$$f(x)=a_rx^{p^r}+a_{r-1}x^{p^{r-1}}+\cdots+a_1x^{p}+a_0x+a$$ be an affine $p$-polynomial over $\gf(q)$. Let $\chi_b$ be a nontrivial additive character of $\gf(q)$ with $b\in \gf(q)^*$. Then
$$\sum_{c\in \gf(q)}\chi_b(f(c))=\left\{\begin{array}{ll}
\chi_b(a)q    &   \mbox{ if }ba_r+b^pa_{r-1}^p+\cdots+b^{p^{r-1}}a_{1}^{p^{r-1}}+b^{p^{r}}a_{0}^{p^{r}}=0,\\
0    &   \mbox{ otherwise. }
\end{array} \right.$$
\end{lemma}

Let $q=2^m$. The value of another class of Weil sums defined by
$$S_{h}(a,b)=\sum_{x\in \gf(q)}\chi_1(f(x)).$$
When $f(x)=ax^{2^h+1}+bx$, $a,b\in \gf(q)$, this sum was determined by Coulter in 1999
and is given below.

\vspace{.1cm}

\begin{lemma}\label{lemma-coulter}\cite{C}
Let $q=2^m$, $\alpha$ be a primitive element of $\gf(q)$ and $f(x)=ax^{2^h+1}+bx$, $a,b\in \gf(q)$. Let $e=\gcd(m,h)$ and $a\in \gf(q)^*$.
\begin{enumerate}
\item Let $m/e$ be odd. Then $S_{h}(a,0)=0$. If  $b\in \gf(q)^*$, then $S_{h}(a,b)=S_{h}(1,bc^{-1})$ where $c\in \gf(q)^*$ is the unique element satisfying $c^{2^h+1}=a$ and
    $$S_{h}(1,b)=\left\{\begin{array}{ll}
0    &   \mbox{ if }\tr_{2^m/2^e}(b)\neq 1,\\
\pm 2^{\frac{m+e}{2}}   &   \mbox{ if }\tr_{2^m/2^e}(b)=1.
\end{array} \right.$$
\item Let $m/e$ be even. If $b=0$, we have
$$S_{h}(a,0)=\left\{\begin{array}{ll}
-(-1)^{\frac{m}{2e}} 2^{\frac{m}{2}+e}    &   \mbox{ if }a=\alpha^{t(2^e+1)}\mbox{ for some integer }t,\\
(-1)^{\frac{m}{2e}} 2^{\frac{m}{2}}   &   \mbox{ if }a\neq\alpha^{t(2^e+1)}\mbox{ for any integer }t.
\end{array} \right.$$
If $b\in \gf(q)^*$, there are two cases as follows.
\begin{enumerate}
\item If $a\neq\alpha^{t(2^e+1)}$ for any integer $t$, then $g(x)=a^{2^h}x^{2^{2h}}+ax$ is a permutation polynomial. Let $x_0\in \gf(q)$ be the unique solution of $g(x_0)=b^{2^h}$. Then
    $$S_h(a,b)=(-1)^{\frac{m}{2e}}2^{\frac{m}{2}}\chi_1(ax_0^{2^h+1}).$$
\item If $a=\alpha^{t(2^e+1)}$ for some integer $t$, then $S_h(a,b)=0$ unless $g(x)=a^{2^h}x^{2^{2h}}+ax=b^{2^h}$ is solvable. If $g(x_0)=b^{2^h}$ has a solution $x_0$, then
    $$S_{h}(a,0)=\left\{\begin{array}{ll}
-(-1)^{\frac{m}{2e}}2^{\frac{m}{2}+e}\chi_1(ax_0^{2^h+1})   &   \mbox{ if }\tr_{2^m/2^e}\neq 0,\\
(-1)^{\frac{m}{2e}}2^{\frac{m}{2}}\chi_1(ax_0^{2^h+1})   &   \mbox{ if }\tr_{2^m/2^e}\neq 0.
\end{array} \right.$$
\end{enumerate}
\end{enumerate}
\end{lemma}

We will need these lemmas in later sections.

\subsection{Subfield codes}

Let $q$ be a power of a prime and $m$ a positive integer. Let $\C$ be an $[n,k]$ linear code over the finite field $\gf(q^m)$. Now we construct a new $[n, k']$ code $\C^{(q)}$ over $\gf(q)$ as follows. Let $G$ be a generator matrix of $\C$. Take a basis of $\gf(q^m)$ over $\gf(q)$. Represent each entry of $G$ as an $m \times 1$ column vector of $\gf(q)^m$ with respect to this basis, and replace each entry of $G$ with the corresponding $m \times 1$ column vector of $\gf(q)^m$. With this method, $G$ is modified into a $km \times n$ matrix over $\gf(q)$ generating the new \emph{subfield code} $\C^{(q)}$ over $\gf(q)$ with length $n$.
It is known that the subfield code $\C^{(q)}$ is independent of both the choice of the basis of $\gf(q^m)$ over $\gf(q)$ and the choice of the generator matrix $G$ of $\C$ (see Theorems 2.1 and 2.6 in \cite{DH}).

By definition, the dimension $k'$ of $\C^{(q)}$ satisfies $k'\leq mk$. To the best of our knowledge, the only references
on subfield codes are \cite{CCD, CCZ, DH, HD}. Recently, some basic
results about subfield codes were derived and the subfield codes of ovoid codes were studied
in \cite{DH}. It was demonstrated that the subfield codes of ovoid codes are very attractive
\cite{DH}. The parameters of some hyperoval codes and the conic codes were also studied in \cite{HD}.

For a linear code $\C$ over $\gf(q^m)$, we denote by $\C^{\bot}$ the dual code of $\C$. A relationship between the minimal distance of $\C^{\bot}$  and that of $\C^{(q)\bot}$ is given as follows.
\begin{lemma}\label{th-dualdistance}\cite[Theorem 2.7]{DH}
The minimal distance $d^\perp$ of $\C^{\bot}$ and the minimal distance $d^{(q)\perp}$ of $\C^{(q)\bot}$ satisfy
$$d^{(q)\perp}\geq d^\perp.$$
\end{lemma}

The trace representation of the $q$-ary subfield code $\C^{(q)}$ of a linear code $\C$ over $\gf(q^m)$  is presented as follows.
\begin{lemma}\label{th-tracerepresentation}\cite[Theorem 2.5]{DH}
Let $\C$ be an $[n,k]$ code over $\gf(q^m)$. Let $G=[g_{ij}]_{1\leq i \leq k, 1\leq j \leq n}$ be a generator matrix of $\C$. Then the trace representation of the subfield code $\C^{(q)}$ is given by
$$
\C^{(q)}=\left\{\left(\tr_{q^m/q}\left(\sum_{i=1}^{k}a_ig_{i1}\right),
\cdots,\tr_{q^m/q}\left(\sum_{i=1}^{k}a_ig_{in}\right)\right):a_1,\ldots,a_k\in \gf(q^m)\right\}.
$$
\end{lemma}

The subfield subcode $\C|_{\gf(q)}$ of an $[n, k]$ code over $\gf(q^m)$ is the set of codewords
in $\C$ each of whose components is in $\C$. Hence, the dimension of the subfield subcode $\C|_{\gf(q)}$ is at most $k$. Thus, the subfield code over $\gf(q)$ and subfield subcode over $\gf(q)$ of a linear code over $\gf(q^m)$ are different codes in general. In fact, it is easy
to see that the subfield subcode $\C|_{\gf(q)}$ is a subcode of the subfield code $\C^{(q)}$.
Subfield codes were considered
in \cite{CCZ} and \cite{CCD} without using the name ``subfield codes". Subfield codes were
defined formally in \cite[p. 5117]{MagmaHK} and a Magma function for subfield codes is implemented in the Magma package. Notice that subfield subcodes were well studied in the literature \cite{Dels,DK02,GP10,PJ14,Stich}.

\subsection{Linear codes from maximal arcs in $\PG(2, \gf(2^m))$}

Let $q$ be a power of a prime. A \emph{maximal} $(n,h)$-\emph{arc} in the Desarguessian projective plane $\PG(2,\gf(q))$ is a set of $n=hq+h-q$ points such that every line meets $\cA$ in just $h$ points or in none at all. A line is called a \emph{secant} if it meets $\cA$, and \emph{external line} otherwise. For a maximal $(n,h)$-arc $\cA$, the set of lines external to $\cA$ is a maximal $(q(q-h+1)/h,q/h)$-arc in the dual plane and called the \emph{dual} of $\cA$. It follows that a necessary condition for a maximal $(n,h)$-arc exists is $h|q$. Any point of $\PG(2,\gf(q))$ is a $(1,1)$-arc and the complement of any line is a maximal $(q^2,q)$-arc, which are called \emph{trivial} maximal arcs. In \cite{BBM}, Ball, Blokhuis and Mazzocca proved that no nontrivial maximal arcs exist in $\PG(2,\gf(q))$ for odd $q$. When $h=2$, maximal arcs become hyperovals.

In 1969, Denniston used a special pencil  of conics to construct maximal arcs in $\PG(2,\gf(q))$ for even $q$ \cite{D}. Let $x^2+\beta x+1$ be irreducible over $\gf(q)$. Define
$$
F_\lambda:=\{(x,y,z):\lambda x^2+y^2+\beta yz +z^2=0\},\ \lambda\in \gf(q)\cup \{\infty\}.
$$
It is easy to verify that $F_0=\{(1,0,0)\}$ and $F_\infty$ is the line $x=0$.
Each other $F_{\lambda}$ is a conic for $\lambda \in \gf(q)^*$.
We call $F_\lambda$ the \emph{standard pencil} \cite{M}. The following theorem documents \emph{Denniston arcs}.

\begin{theorem}\cite{AL, D}\label{th-1}
Let $H$ be a subset of $\gf(q)$ of order $h$. Then the set $\cA:=\cup_{\lambda\in H}F_\lambda$ is a maximal $(n,h)$-arc if and only if $H$ is an additive subgroup of $\gf(q)$, where $n=hq+h-q$.
\end{theorem}

Given a maximal $(n,h)$-arc $\cA$, the points in the arc define a $3\times n$ matrix $G$ over $\gf(q)$ with each column vector of $G$ being a point in $\cA$. Let $\C(\cA)$ be the linear code spanned by the rows of $G$. Then $\C(\cA)$ is referred to as a \emph{maximal arc code}. By definition, $\cA$ meets each line in either 0 or $h$ points. Note that in $\PG(2,\gf(q))$ lines and hyperplanes are the same. Then it is easy to derive the weight distribution of $\C(\cA)$ and the parameters of $\C(\cA)^{\perp}$ given in the following theorem (see, for example,  \cite[Theorem 6]{WDT}).

\begin{theorem}\label{th-2}
Let $q=2^m$ for any $m\geq 2$ and $h=2^s$ with $1\leq s < m$. Let $\cA$ be a maximal $(n,h)$-arc in $\PG(2,\gf(q))$. Then the maximal arc code $\C(\cA)$ has parameters $[n,3,n-h]$ and weight enumerator
$$1+\frac{(q^2-1)n}{h}z^{n-h}+\frac{(q^3-1)h-(q^2-1)n}{h}z^n,$$ where $n=hq+h-q$. The dual $\C(\cA)^{\perp}$ has parameters $[n,n-3,4]$ if $s=1$ and $[n,n-3,3]$ if $s>1$.
\end{theorem}

We will use the Denniston arc codes to construct a class of distance-optimal binary codes
in Section \ref{sec-232}.

\subsection{Linear codes from maximal arcs in $\PG(r, \gf(q))$}

An \emph{arc}\index{arc} in $\PG(r, \gf(q))$ is a set of at least $r+1$ points
in $\PG(r, \gf(q))$ such that no $r+1$ of them lie in a hyperplane. A \emph{cap}\index{cap} in
$\PG(r, \gf(q))$ is a set of points such that no three are collinear.

Given a set $\cA=\{\bg_1, \bg_2, \cdots, \bg_{n}\}$ with $n$ points in $\PG(r, \gf(q))$,
where each $\bg_i$ is a $(r+1) \times 1$ vector in $\gf(q)^{r+1}$, we define a matrix
\begin{eqnarray}\label{eqn-genmatA}
G_{\cA}=[\bg_1 \bg_2 \cdots \bg_n].
\end{eqnarray}
The linear code over $\gf(q)$ with generator matrix $G_{\cA}$ is denoted by $\C(\cA)$.
The following theorem is well known (see, for example,
\cite[Chapter 12]{Dingbk18}).

\begin{theorem}\label{th-arccode}
Let $\cA$ be an $n$-subset of the point set in $\PG(r, \gf(q))$ with $n \geq r+1$.
Then $\cA$ is an arc in $\PG(r, \gf(q))$ if and only if the corresponding code $\C_{\cA}$
is an $[n, r+1, n-r]$ MDS code over $\gf(q)$.
\end{theorem}

\vspace{.1cm}
If $\cA$ is an arc in  $\PG(r, \gf(q))$, the code $\C_{\cA}$ is called an \emph{arc code}.
In Section \ref{sec-june23}, we will use some arc codes over $\gf(2^m)$ to construct a class
of distance-optimal binary linear codes.

\section{The class of binary codes with parameters $[2^{m+s}+2^s-2^m,2^{m+s}+2^s-2^m-2m-2,4]$}\label{sec-232}

In this section, we present our first class of distance-optimal binary linear codes which are
based on the Denniston arcs.

\subsection{The construction of the binary codes}

Let $q=2^m$ with $m\geq 2$ and $\gf(q)^*:=\gf(q)\setminus \{0\}$. Let $H$ be an additive subgroup of $\gf(q)$ with $h:=|H|=2^s$, where $1< s <m$. Let $x^2+\beta x+1$ be irreducible over $\gf(q)$. Denote $H^*=H\setminus\{0\}$. Recall the standard pencil in $\PG(2,\gf(q))$ defined by
$$F_\lambda:=\{(x,y,z):\lambda x^2+y^2+\beta yz +z^2=0\},\ \lambda\in \gf(q)\cup \{\infty\}.$$
By Theorem \ref{th-1}, the set $\cA:=\cup_{\lambda\in H}F_\lambda$ is a maximal $(n,h)$-arc called the Denniston arc, where $n=hq+h-q$.

\vspace{0.1cm}

\begin{lemma}\label{lem-standardpencil}
If $\lambda=0$, then $F_0=\{(1,0,0)\}$. If $\lambda\in \gf(q)^*$, then
$$F_\lambda=\left\{(\lambda^{-\frac{q}{2}},1,0)\right\} \bigcup \left\{\left(\lambda^{-\frac{q}{2}}(y+\beta^{\frac{q}{2}}y^{\frac{q}{2}}+1),y,1\right):y\in \gf(q)\right\}.$$
\end{lemma}
\begin{proof}
Since $x^2+\beta x+1$ is irreducible over $\gf(q)$, we have $F_0=\{(1,0,0)\}$. If $\lambda\in \gf(q)^*$, then $\lambda x^2+y^2+\beta yz +z^2=0$ implies
$$x=\lambda^{-\frac{q}{2}}(y+\beta^{\frac{q}{2}}y^{\frac{q}{2}}z^{\frac{q}{2}}+z).
$$
Since $F_\lambda$ is a subset of the point set of $\PG(2,\gf(q))$ for $\lambda \in \gf(q)^*$, we have
\begin{eqnarray*}F_\lambda &=&\left\{\left(\lambda^{-\frac{q}{2}}(y+\beta^{\frac{q}{2}}y^{\frac{q}{2}}z^{\frac{q}{2}}+z),y,z\right):y,z \in \gf(q)\right\}\\
&=&\left\{(\lambda^{-\frac{q}{2}},1,0)\right\} \bigcup \left\{\left(\lambda^{-\frac{q}{2}}\left(\frac{y}{z}+\beta^{\frac{q}{2}}\left(\frac{y}{z}\right)^{\frac{q}{2}}+1\right),\frac{y}{z},1\right):y\in \gf(q),z\in \gf(q)^*\right\}\\
&=&\left\{(\lambda^{-\frac{q}{2}},1,0)\right\} \bigcup \left\{\left(\lambda^{-\frac{q}{2}}(y+\beta^{\frac{q}{2}}y^{\frac{q}{2}}+1),y,1\right):y\in \gf(q)\right\}.
\end{eqnarray*}
The proof is completed.
\end{proof}

Let $H=\{0,\lambda_1,\lambda_2,\cdots,\lambda_{h-1}\}$ and $\gf(q)=\{y_1,\cdots,y_q\}$. Define
$$G_\lambda=\begin{bmatrix} \lambda_1^{-\frac{q}{2}} & \cdots & \lambda_{h-1}^{-\frac{q}{2}}\\
 1 & \cdots & 1\\
0 & \cdots & 0 \\\end{bmatrix}$$
and
\begin{eqnarray}\label{matrix-1}
G_{\lambda}(y)=\begin{bmatrix}  \lambda_1^{-\frac{q}{2}}(y+\beta^{\frac{q}{2}}y^{\frac{q}{2}}+1) & \cdots & \lambda_{h-1}^{-\frac{q}{2}}(y+\beta^{\frac{q}{2}}y^{\frac{q}{2}}+1)\\
  y & \cdots & y\\
1 & \cdots & 1 \\
\end{bmatrix}
\end{eqnarray}
for $y\in \gf(q)$.
By Lemma \ref{lem-standardpencil}, the Denniston arc code $\C(\cA)$ has a generator matrix
$$G_{\cA}=\begin{bmatrix}G_\lambda & G_\lambda(y_1) & \cdots &  G_\lambda(y_q) & \begin{array}{c}
1 \\
0\\
0\\\end{array}\end{bmatrix}.$$
Due to Theorem \ref{th-2}, the Denniston arc code $\C(\cA)$  has parameters $[n,3,n-h]$, where $n=hq+h-q$. Consider now the \emph{augmented Denniston arc code} defined by
\begin{eqnarray}\label{defn}
\widetilde{\C}(\cA)=\{\bc+b\textbf{1}:\bc\in \C(\cA),b\in \gf(q)\},
\end{eqnarray}
where $\textbf{1}$ denotes the vector $(1, 1, \cdots, 1) \in \gf(q)^n$.
Then $\widetilde{\C}(\cA)$ has a generator matrix
\begin{eqnarray}\label{matrix-2}
\widetilde{G}_{\cA}=\begin{bmatrix}G_\lambda & G_\lambda(y_1) & \cdots &  G_\lambda(y_q) & \begin{array}{c}
1 \\
0\\
0\\\end{array}\\ \hline
1 & 1 & \cdots & 1 & 1\end{bmatrix}.\end{eqnarray}
It is obvious that $\widetilde{\C}(\cA)$ is an $[n,4]$ linear code over $\gf(q)$.

Now we consider the binary subfield code $\widetilde{\C}(\cA)^{(2)}$ of $\widetilde{\C}(\cA)$ defined in Equation (\ref{defn}). Combining Lemma \ref{th-tracerepresentation} and Equation (\ref{matrix-2}) yields the trace representation of $\widetilde{\C}(\cA)^{(2)}$ as follows.

\begin{lemma}\label{lem-trace}
The trace representation of $\widetilde{\C}(\cA)^{(2)}$ is given by
\begin{eqnarray*}
\lefteqn{ \widetilde{\C}(\cA)^{(2)}= } \\
& &\myatop{\bigg\{\bigg(\left(\tr_{q/2}(a_1\lambda^{-\frac{q}{2}}+a_2)+c\right)_{\lambda\in H^*},\left(\tr_{q/2}(a_1\lambda^{-\frac{q}{2}}(y+\beta^{\frac{q}{2}}y^{\frac{q}{2}}+1)+a_2y)+b+c\right)_{\myatop{\lambda\in H^*,}{y\in \gf(q)}},}{\tr_{q/2}(a_1)+c\bigg):
a_1,a_2,\in\gf(q),b,c\in \gf(2)\bigg\}.}
\end{eqnarray*}
\end{lemma}

To determine the dimension of $\widetilde{\C}(\cA)^{(2)}$, we need the lemma below.

\vspace{.1cm}

\begin{lemma}\label{lem-number}
Let $q=2^m$ with $m\geq 2$. Let $H$ be an additive subgroup of $\gf(q)$ with $h=|H|=2^s$ and $1\leq s \leq m$. Let $x^2+\beta x+1$ be irreducible over $\gf(q)$. Denote
$$N=\sharp\left\{(\lambda,y)\in H^*\times \gf(q):\tr_{q/2}\left(A_1\lambda^{-\frac{q}{2}}(y+\beta^{\frac{q}{2}}y^{\frac{q}{2}})+A_2(y+1)\right)+B=0\right\},$$
where $A_1,A_2\in \gf(q)$ and $B\in \gf(2)$. Then $N=q(h-1)$ if and only if $(A_1,A_2,B)=(0,0,0)$.
\end{lemma}

\begin{proof}
If $(A_1,A_2,B)=(0,0,0)$, then it is clear that $N=q(h-1)$.

In the following, we assume that $N=q(h-1)$. Our goal is to prove $(A_1,A_2,B)=(0,0,0)$.
Let $\chi$ be the canonical additive character of $\gf(q)$. By the orthogonality relation of additive characters, we have
\begin{eqnarray}\label{proof-eqn-1}
\nonumber N&=&\sharp\left\{(\lambda,y)\in H^*\times \gf(q):\tr_{q/2}\left(A_1\lambda^{-\frac{q}{2}}\beta^{\frac{q}{2}}y^{\frac{q}{2}}+(A_2+A_1\lambda^{-\frac{q}{2}})y+A_2\right)+B=0\right\}\\
\nonumber&=&\frac{1}{2}\sum_{\lambda\in H^*}\sum_{y\in \gf(q)}\sum_{z\in \gf(2)}(-1)^{z\left(\tr_{q/2}\left(A_1\lambda^{-\frac{q}{2}}\beta^{\frac{q}{2}}y^{\frac{q}{2}}+(A_2+A_1\lambda^{-\frac{q}{2}})y+A_2\right)+B\right)}\\
&=&\frac{q}{2}(h-1)+\frac{(-1)^B}{2}\sum_{\lambda\in H^*}\sum_{y\in \gf(q)}\chi(A_1\lambda^{-\frac{q}{2}}\beta^{\frac{q}{2}}y^{\frac{q}{2}}+(A_2+A_1\lambda^{-\frac{q}{2}})y+A_2).
\end{eqnarray}
By Lemma \ref{lem-p-polynomial},
\begin{eqnarray}\label{proof-eqn-2}
\nonumber \Omega(A_1,A_2)&:=&\sum_{\lambda\in H^*}\sum_{y\in \gf(q)}\chi(A_1\lambda^{-\frac{q}{2}}\beta^{\frac{q}{2}}y^{\frac{q}{2}}+(A_2+A_1\lambda^{-\frac{q}{2}})y+A_2)\\
&=&q\chi(A_2)N_\lambda,
\end{eqnarray}
where
$$N_\lambda:=\sharp \left\{\lambda\in H^*:A_1\lambda^{-\frac{q}{2}}\beta^{\frac{q}{2}}+(A_2+A_1\lambda^{-\frac{q}{2}})^{\frac{q}{2}}=0\right\}.$$
By Equations (\ref{proof-eqn-1}) and (\ref{proof-eqn-2}), we have
$$q(h-1)=\frac{q}{2}(h-1)+\frac{(-1)^B}{2}q\chi(A_2)N_\lambda,$$
which is equivalent to
$N_\lambda=h-1$ and $\tr_{q/2}(A_2)+B=0$. For $\lambda\in H^*$, the equation
$$A_1\lambda^{-\frac{q}{2}}\beta^{\frac{q}{2}}+(A_2+A_1\lambda^{-\frac{q}{2}})^{\frac{q}{2}}=0$$
is equivalent to
$$\left(A_1\lambda^{-\frac{q}{2}}\beta^{\frac{q}{2}}+(A_2+A_1\lambda^{-\frac{q}{2}})^{\frac{q}{2}}\right)^2=0,$$ i.e.,
\begin{eqnarray}\label{eqn-2-1}
A_1^2\lambda^{-1}\beta+A_2+A_1\lambda^{-\frac{q}{2}}=0.\end{eqnarray} It is clear that Equation (\ref{eqn-2-1}) is equivalent to
$$\left(A_1^2\lambda^{-1}\beta+A_2+A_1\lambda^{-\frac{q}{2}}\right)^2=0,$$ i.e.,
$$A_1^4\lambda^{-2}\beta^2+A_2^2+A_1^2\lambda^{-1}=0,$$
which can be written as
$$A_2^2\lambda^2+A_1^2\lambda+A_1^4\beta^2=0.$$ Hence we have
$$N_\lambda=\sharp \left\{\lambda\in H^*:A_2^2\lambda^2+A_1^2\lambda+A_1^4\beta^2=0\right\}.$$
Let $\lambda_1,\lambda_2\in H^*$ and $\lambda_1\neq \lambda_2$ such that
\begin{eqnarray}\label{proof-eqn-3}A_2^2\lambda_1^2+A_1^2\lambda_1+A_1^4\beta^2=0,\end{eqnarray} and
\begin{eqnarray}\label{proof-eqn-4}A_2^2\lambda_2^2+A_1^2\lambda_2+A_1^4\beta^2=0.
\end{eqnarray}
Since $H$ is an additive subgroup, we have $\lambda_1+\lambda_2\in H^*$. Then $N_\lambda=h-1$ yields that
\begin{eqnarray}\label{proof-eqn-5}A_2^2(\lambda_1+\lambda_2)^2+A_1^2(\lambda_1+\lambda_2)+A_1^4\beta^2=0.
\end{eqnarray}
Combining Equations (\ref{proof-eqn-3}), (\ref{proof-eqn-4}) and (\ref{proof-eqn-5}) yields $A_1=0$ as $\beta\neq 0$. Then $A_2=0$ as $N_\lambda=h-1>0$ and $\lambda\in H^*$.
We also have $B=0$ as $\tr_{q/2}(A_2)+B = 0$.
According to the preceding discussions, we have proved $(A_1,A_2,B)=(0,0,0)$.
Then the desired conclusion follows.
\end{proof}

\begin{lemma}\label{lem-dimension}
Let $q=2^m$ with $m\geq 2$. Let $H$ be an additive subgroup of $\gf(q)$ with $h=|H|=2^s$ and $1< s <m$. Let $x^2+\beta x+1$ be irreducible over $\gf(q)$. Then the dimension of $\widetilde{\C}(\cA)^{(2)}$ is $2m+2$.
\end{lemma}

\begin{proof}
By Lemma \ref{lem-trace}, we assume that there exist four-tuples $(a_1,a_2,b,c)\in \gf(q)\times \gf(q)\times \gf(2)\times \gf(2)$ and $(a_1',a_2',b',c')\in \gf(q)\times \gf(q)\times \gf(2)\times \gf(2)$ such that
\begin{eqnarray*}
& &\myatop{\bigg(\left(\tr_{q/2}(a_1\lambda^{-\frac{q}{2}}+a_2)+c\right)_{\lambda\in H^*},\left(\tr_{q/2}(a_1\lambda^{-\frac{q}{2}}(y+\beta^{\frac{q}{2}}y^{\frac{q}{2}}+1)+a_2y)+b+c\right)_{\myatop{\lambda\in H^*,}{y\in \gf(q)}},}{\tr_{q/2}(a_1)+c\bigg)}\\
& &\myatop{=\bigg(\left(\tr_{q/2}(a_1'\lambda^{-\frac{q}{2}}+a_2')+c'\right)_{\lambda\in H^*},\left(\tr_{q/2}(a_1'\lambda^{-\frac{q}{2}}(y+\beta^{\frac{q}{2}}y^{\frac{q}{2}}+1)+a_2'y)+b'+c'\right)_{\myatop{\lambda\in H^*,}{y\in \gf(q)}},}{\tr_{q/2}(a_1')+c'\bigg).}
\end{eqnarray*}
Then we have
\begin{eqnarray}\label{eqn-1}
\tr_{q/2}\left((a_1+a_1')\lambda^{-\frac{q}{2}}+a_2+a_2'\right)+c+c'=0
\end{eqnarray}
for any $\lambda\in H^*$,
\begin{eqnarray}\label{eqn-2}
\tr_{q/2}\left((a_1+a_1')\lambda^{-\frac{q}{2}}(y+\beta^{\frac{q}{2}}y^{\frac{q}{2}}+1)+(a_2+a_2')y\right)+b+b'+c+c'=0
\end{eqnarray}
for any $(\lambda,y)\in H^*\times \gf(q)$, and
\begin{eqnarray}\label{eqn-3}
\tr_{q/2}(a_1+a_1')+c+c'=0.
\end{eqnarray}
Combining Equations (\ref{eqn-1}) and (\ref{eqn-3}) yields
\begin{eqnarray}\label{eqn-4}
\tr_{q/2}\left((a_1+a_1')(\lambda^{-\frac{q}{2}}+1)+a_2+a_2'\right)=0
\end{eqnarray}
for any $\lambda\in H^*$.
Combining Equations (\ref{eqn-2}) and (\ref{eqn-3}) yields
\begin{eqnarray}\label{eqn-5}
\tr_{q/2}\left((a_1+a_1')\left(\lambda^{-\frac{q}{2}}(y+\beta^{\frac{q}{2}}y^{\frac{q}{2}}+1)+1\right)+(a_2+a_2')y\right)+b+b'=0
\end{eqnarray}
for any $(\lambda,y)\in H^*\times \gf(q)$.
Then Equations (\ref{eqn-4}) and (\ref{eqn-5}) imply that
\begin{eqnarray*}
\tr_{q/2}\left((a_1+a_1')\lambda^{-\frac{q}{2}}(y+\beta^{\frac{q}{2}}y^{\frac{q}{2}})+(a_2+a_2')(y+1)\right)+b+b'=0
\end{eqnarray*}
for any $(\lambda,y)\in H^*\times \gf(q)$. By Lemma \ref{lem-number}, we deduce that $a_1=a_1'$, $a_2=a_2'$ and $b=b'$. Then Equation (\ref{eqn-3}) implies $c=c'$. Thus the dimension of $\widetilde{\C}(\cA)^{(2)}$ is $2m+2$.
\end{proof}

Denote by $\tilde{d}^{\perp}$ and $\tilde{d}^{(2)\perp}$ the minimal distances of $\widetilde{\C}(\cA)^{\perp}$ and $\widetilde{\C}(\cA)^{(2)\perp}$, respectively. The following theorem is the main result of this section.

\begin{theorem}
Let $q=2^m$ with $m\geq 2$. Let $H$ be an additive subgroup of $\gf(q)$ with $h=|H|=2^s$ and $1< s <m$. Let $x^2+\beta x+1$ be irreducible over $\gf(q)$. Then the dual $\widetilde{\C}(\cA)^{(2)\perp}$ of $\widetilde{\C}(\cA)^{(2)}$ is distance-optimal with respect to the sphere-packing bound and has parameters $$[2^{m+s}+2^s-2^m,2^{m+s}+2^s-2^m-2m-2,4].$$
\end{theorem}

\begin{proof}
By Lemma \ref{lem-dimension}, the dimension of $\widetilde{\C}(\cA)^{(2)\perp}$ equals $n-2m-2=2^{m+s}+2^s-2^m-2m-2$.

Consider the matrix $G_\lambda(y)$ in Equation (\ref{matrix-1}), we claim that $y+\beta^{\frac{q}{2}}y^{\frac{q}{2}}+1\neq 0$ for any $y\in \gf(q)$. Otherwise, if $y+\beta^{\frac{q}{2}}y^{\frac{q}{2}}+1= 0$ for some $y\in \gf(q)$, we have
$y^2+\beta y+1=0$ for some $y\in \gf(q)$, which contradicts with our assumption that $x^2+\beta x+1$ is irreducible over $\gf(q)$. Hence, any two columns of the matrix $\widetilde{G}_{\cA}$ in Equation (\ref{matrix-2}) are different. Since the matrix in Equation (\ref{matrix-2}) is a parity-check matrix of $\widetilde{\C}(\cA)^{\perp}$, we deduce that $\tilde{d}^{\perp}\geq 3$. It is clear that $\widetilde{\C}(\cA)$ contains the codeword $(1,1,\cdots,1)$. Hence all the weights of $\widetilde{\C}(\cA)^{\perp}$ are even. Then we deduce that $\tilde{d}^{\perp}\geq 4$. By Lemma \ref{th-dualdistance}, we obtain that $\tilde{d}^{(2)\perp}\geq\tilde{d}^{\perp}\geq 4$.
By the sphere-packing bound of binary codes, we have
\begin{eqnarray}\label{eqn-bound}
2^{2^{m+s}+2^s-2^m}\geq 2^{2^{m+s}+2^s-2^m-2m-2}\sum_{i=0}^{\lfloor\frac{\tilde{d}^{(2)\perp}-1}{2}\rfloor}\binom{2^{m+s}+2^s-2^m}{i},
\end{eqnarray}
where $\lfloor x\rfloor$ denotes the floor function. Suppose that $\tilde{d}^{(2)\perp}=5$. Then Equation (\ref{eqn-bound}) becomes
\begin{eqnarray*}2^{2m+2}&\geq & 1+2^{m-1}+2^{s-1}+2^{2m+s}(2^{s-1}-1)+2^{m+s-1}(2^{s+1}-1)+2^{2m+1}+2^{2s-1}\\
&>&2^{2m+s}(2^{s-1}-1)\end{eqnarray*}
which is a contradiction as $s>1$. Hence $\tilde{d}^{(2)\perp}\leq 4$.
Then we deduce that $\tilde{d}^{(2)\perp}=4$ and the desired conclusion follows.
\end{proof}

Below we present an example of the codes treated before.

\vspace{.1cm}

\begin{example}
Let $m=5$ and $w$ be a generator of $\gf(2^5)^*$ with $w^5+w^2+1=0$. Let the subgroup $\cA$
of $(\gf(32), +)$ be $\{0, 1, w^{11}, w, w^2, w^5, w^{18}, w^{19}\}$. Hence, $h=|\cA|=8$. Then the Denniston arc
code $\C(\cA)$ over $\gf(32)$ has parameters $[232, 3, 224]$ and weight enumerator
$$
1 + 29667 z^{224} + 3100 z^{232}.
$$
Its dual $\C(\cA)^\perp$ over $\gf(32)$ has parameters $[232, 229, 3]$.


The subfield code  $\widetilde{\C}(\cA)^{(2)}$ over $\gf(2)$ has parameters $[232, 12, 8]$.
Its dual $\widetilde{\C}(\cA)^{(2)\perp}$ has parameters $[232, 220, 4]$.
\end{example}

It is conjectured that $\widetilde{\C}(\cA)^{(2)}$ has minimal distance $h$ which is confirmed by our computer experiments. Note that although the binary code $\widetilde{\C}(\cA)^{(2)}$ has poor error-correcting
capability, the dual code $\widetilde{\C}(\cA)^{(2)\perp}$ is distance-optimal. The class
of codes $\widetilde{\C}(\cA)^{(2)\perp}$ achieves the first objective of this paper.

\subsection{A comparison of $\widetilde{\C}(\cA)^{(2)\perp}$ with the extended binary Hamming code}

The binary Hamming code has parameters $[2^m-1, 2^m-1-m, 3]$ and is perfect in the sense
that it meets the sphere-packing bound.
The extended binary Hamming code has parameters $[2^m, 2^m-m-1, 4]$ and is distance-optimal.
The information rate of this code is
$$
R_1:=\frac{2^m-m-1}{2^m}.
$$
The information rate of the binary code $\widetilde{\C}(\cA)^{(2)\perp}$ is
$$
R_2:=\frac{2^{m+s}+2^s-2^m-2m-2}{2^{m+s}+2^s-2^m}.
$$
Since the two codes have the same minimum distance $4$, we can compare their information rates.
When $s \geq 2$ and $m \geq 4$, it
can be verified that the information rate of $\widetilde{\C}(\cA)^{(2)\perp}$ is larger than
that of the extended binary Hamming code, i.e., $R_2>R_1$. Although both codes are distance-optimal, the code
$\widetilde{\C}(\cA)^{(2)\perp}$ developed in this paper is better than the extended binary
Hamming code in terms of their information rates. Hence, the class of binary codes $\widetilde{\C}(\cA)^{(2)\perp}$ are quite
attractive. In addition, the parameters of the class of binary codes $\widetilde{\C}(\cA)^{(2)\perp}$ look new.

\section{The class of optimal binary codes with parameters $[2^m+2, 2^m-2m, 6]$}\label{sec-june23}

In this section, we present our second class of distance-optimal binary linear codes which
will be constructed with a class of maximal arcs in $\PG(3, \gf(2^m))$.

\subsection{The construction of the binary codes}

Let $q$ be a prime power. It is known that the maximum number of points in an arc in
$\PG(3,\gf(q))$ is $q+1$ \cite{S}. The following lemma documents a known arc with $q+1$
points in $\PG(3,\gf(q))$.

\vspace{0.1cm}
\begin{lemma}\label{lemma-arc}\cite{Hirschfeld, HX}
For $q=2^m$ with $m \geq 2$. The set
$$\cA=\left\{(x^{2^h+1},x^{2^h},x,1):x\in \gf(q)\right\}\cup \{(1,0,0,0)\}$$ is an arc in $\PG(3,\gf(q))$ if and only if $\gcd(m,h)=1$.
\end{lemma}

\vspace{0.1cm}

Let $\gf(q)=\{x_1,x_2,\cdots,x_q\}$. Let $\cA$ be the arc defined in Lemma \ref{lemma-arc} with $\gcd(m,h)=1$. Define
\begin{eqnarray}\label{eqn-generatormatrix}
G_{\cA}=\begin{bmatrix}x_1^{2^h+1} & x_2^{2^h+1} & \cdots &  x_q^{2^h+1} & 1\\
x_1^{2^h} & x_2^{2^h} & \cdots &  x_q^{2^h} & 0\\
x_1 & x_2 & \cdots &  x_q & 0\\
1 & 1 & \cdots & 1 & 0\\\end{bmatrix}.\end{eqnarray}
Let $\C(\cA)$ be the arc  code with the generator matrix $G_{\cA}$.

Combining Theorem \ref{th-arccode} and Lemma \ref{lemma-arc} directly yields the following result.

\vspace{0.1cm}

\begin{lemma}\label{lemma-parameters of arc code}
Let $q=2^m$ with $m \geq 2$. Let $\C(\cA)$ be the arc  code with the generator matrix $G_{\cA}$ defined in Equation (\ref{eqn-generatormatrix}). Then $\C(\cA)$ is a $q$-ary MDS linear code with parameters $[2^m+1,4,2^m-2]$. Its dual is an MDS code with parameters $[2^m+1,2^m-3,5]$.
\end{lemma}

\vspace{0.1cm}

Now we consider the binary subfield code $\C(\cA)^{(2)}$ of $\C(\cA)$. The trace representation of $\C(\cA)^{(2)}$ is given in the following lemma.

\vspace{0.1cm}

\begin{lemma}\label{tracerepresnt-2}
Let $\cA$ be the arc defined in Lemma \ref{lemma-arc} with $\gcd(m,h)=1$. Let $\C(\cA)$ be the arc  code with the generator matrix $G_{\cA}$ defined in Equation (\ref{eqn-generatormatrix}). Then the trace representation of $\C(\cA)^{(2)}$ is given by
$$\C(\cA)^{(2)}=\left\{\left(\left(\tr_{2^m/2}(ax^{2^h+1}+bx)+c\right)_{x\in \gf(2^m)},\tr_{2^m/2}(a)\right):a,b\in \gf(2^m),c\in \gf(2)\right\}.$$
\end{lemma}

\begin{proof}
Combining Lemma \ref{th-tracerepresentation} and Equation (\ref{eqn-generatormatrix}) yields the following trace representation of $\C(\cA)^{(2)}$:
\begin{eqnarray*}& &\C(\cA)^{(2)}\\
&=&\left\{\left(\left(\tr_{2^m/2}(a_1x^{2^h+1}+a_2x^{2^h}+a_3x)+c\right)_{x\in \gf(2^m)},\tr_{2^m/2}(a_1)\right):a_1,a_2,a_3\in \gf(2^m),c\in \gf(2)\right\}.\end{eqnarray*}
Note that
\begin{eqnarray*}
\lefteqn{ \tr_{2^m/2}(a_1x^{2^h+1}+a_2x^{2^h}+a_3x)+c } \\
&=&\tr_{2^m/2}(a_1x^{2^h+1})+\tr_{2^m/2}(a_2x^{2^h})+\tr_{2^m/2}(a_3x)+c\\
&=&\tr_{2^m/2}(a_1x^{2^h+1})+\tr_{2^m/2}(a_2^{2^{m-h}}x)+\tr_{2^m/2}(a_3x)+c\\
&=&\tr_{2^m/2}(a_1x^{2^h+1})+\tr_{2^m/2}\left((a_2^{2^{m-h}}+a_3)x\right)+c.
\end{eqnarray*}
Let $b:=a_2^{2^{m-h}}+a_3$. It is clear that $b$ runs through $\gf(2^m)$ with $2^m$ times if $(a_2,a_3)$ runs through $\gf(2^m)\times \gf(2^m)$. Then the desired conclusion follows.
\end{proof}


The dimension of $\C(\cA)^{(2)}$ is given in the following lemma.

\vspace{0.1cm}

\begin{lemma}\label{lemma-dimension}
Let $\cA$ be the maximal arc defined in Lemma \ref{lemma-arc} with $\gcd(m,h)=1$. Let $\C(\cA)$ be the arc  code with the generator matrix $G_{\cA}$ defined in Equation (\ref{eqn-generatormatrix}). Then the dimension of $\C(\cA)^{(2)}$ equals $2m+1$.
\end{lemma}

\begin{proof}
By Lemma \ref{th-tracerepresentation}, we assume that there exist three-tuples $(a,b,c)\in \gf(2^m)\times \gf(2^m)\times \gf(2)$ and $(a',b',c')\in \gf(2^m)\times \gf(2^m)\times \gf(2)$ such that
\begin{eqnarray*}
\lefteqn{\left(\left(\tr_{2^m/2}(ax^{2^h+1}+bx)+c\right)_{x\in \gf(2^m)},\tr_{2^m/2}(a)\right)} \\
&=&\left(\left(\tr_{2^m/2}(a'x^{2^h+1}+b'x)+c'\right)_{x\in \gf(2^m)},\tr_{2^m/2}(a')\right).
\end{eqnarray*}
This implies that
\begin{eqnarray}\label{system}
\left\{\begin{array}{ll}
\tr_{2^m/2}\left((a+a')x^{2^h+1}+(b+b')x\right)+c+c'=0\mbox{ for all }x\in \gf(2^m),\\
\tr_{2^m/2}(a+a')=0.
\end{array} \right.
\end{eqnarray}
Let $x=0$ in the first equation in System (\ref{system}), then we have $c=c'$. Hence
\begin{eqnarray}\label{system-1}
\left\{\begin{array}{ll}
\tr_{2^m/2}\left((a+a')x^{2^h+1}+(b+b')x\right)=0\mbox{ for all }x\in \gf(2^m),\\
\tr_{2^m/2}(a+a')=0.
\end{array} \right.
\end{eqnarray}
Denote
$$N(A,B)=\sharp \{x\in \gf(2^m):\tr_{2^m/2}(Ax^{2^h+1}+Bx)=0\},\ A,B\in \gf(q).$$
Let $\chi$ be the canonical additive character of $\gf(2^m)$. If $A=0,B\neq 0$, then $N(A,B)=2^{m-1}$. If $A\neq 0$, then
\begin{eqnarray*}
N(A,B)&=&\frac{1}{2}\sum_{y\in \gf(2)}\sum_{x\in \gf(2^m)}(-1)^{y\tr_{2^m/2}(Ax^{2^h+1}+Bx)}\\
&=&2^{m-1}+\frac{1}{2}\sum_{x\in \gf(2^m)}\chi(Ax^{2^h+1}+Bx)\\
&=&2^{m-1}+\frac{1}{2}S_h(A,B),
\end{eqnarray*}
where $S_h(A,B)$ is defined in Lemma \ref{lemma-coulter}.
Recall that $\gcd(m,h)=1$. By Lemma \ref{lemma-coulter}, if $m$ is odd, then
$$S_h(A,B)\in \{0,\pm 2^{\frac{m+1}{2}}\};$$ if $m$ is even, then
$$S_h(A,B)\in \{\pm 2^{\frac{m}{2}},\pm 2^{\frac{m}{2}+1}\}.$$
Then it is easy to verify that $N(A,B)<2^m$ always holds if $A\neq 0$. Based on these discussions, we deduce that $N(A,B)=2^m$ if and only if $A=B=0$. Thus System (\ref{system-1}) implies that $a=a',b=b'$. Since $a=a',b=b',c=c'$, the dimension of $\C(\cA)^{(2)}$ follows.
\end{proof}

Combining Lemmas \ref{th-dualdistance}, \ref{lemma-parameters of arc code} and \ref{lemma-dimension}, we obtain the parameters of the dual code $\C(\cA)^{(2)\perp}$ of $\C(\cA)^{(2)}$
which are described in the following lemma.

\begin{lemma}\label{lemma-parameters of dual}
Let $\cA$ be the arc defined in Lemma \ref{lemma-arc} with $\gcd(m,h)=1$ and $m\geq 5$. Let $\C(\cA)$ be the arc  code with the generator matrix $G_{\cA}$ defined in Equation (\ref{eqn-generatormatrix}). Then $\C(\cA)^{(2)\perp}$ is a binary linear code with parameters
$$[2^m+1,2^m-2m,d^{(2)\perp}\geq 5].$$
\end{lemma}

Let $\C(\cA)^{(2)\perp}$ be the binary linear code defined in Lemma \ref{lemma-parameters of dual} and $\overline{\C(\cA)^{(2)\perp}}$ be its extended code. In the following, we give the parameters of $\overline{\C(\cA)^{(2)\perp}}$, which is the main result of this section.

\vspace{.1cm}

\begin{theorem}\label{th-main2}
Let $\cA$ be the arc defined in Lemma \ref{lemma-arc} with $\gcd(m,h)=1$ and $m \geq 5$. Let $\C(\cA)$ be the arc  code with the generator matrix $G_{\cA}$ defined in Equation (\ref{eqn-generatormatrix}). Then the extended code $\overline{\C(\cA)^{(2)\perp}}$ is distance-optimal with respect to the sphere-packing bound and has parameters
$$[2^m+2,2^m-2m,6].$$
\end{theorem}

\begin{proof}
Let $\overline{d^{(2)\perp}}$ denote the minimal distance of $\overline{\C(\cA)^{(2)\perp}}$. It follows from Lemma \ref{lemma-parameters of dual} that $\C(\cA)^{(2)\perp}$ has parameters
$$[2^m+1,2^m-2m,d^{(2)\perp}\geq 5].$$ Since its extended code $\overline{\C(\cA)^{(2)\perp}}$ has only even Hamming weights, we deduce that $\overline{\C(\cA)^{(2)\perp}}$ has parameters
$$[2^m+2,2^m-2m,\overline{d^{(2)\perp}}\geq 6].$$ By the sphere-packing bound, we have
$$2^{2^m+2}\geq 2^{2^m-2m}\sum_{i=0}^{\left\lfloor\frac{\overline{d^{(2)\perp}}-1}{2}\right\rfloor}\binom{2^m+2}{i},$$
i.e.,
$$2^{2m+2}\geq \sum_{i=0}^{\left\lfloor\frac{\overline{d^{(2)\perp}}-1}{2}\right\rfloor}\binom{2^m+2}{i}.$$
Since $m\geq 5$, it is easy to deduce that $\overline{d^{(2)\perp}}\leq 6$. Thus $\overline{d^{(2)\perp}}= 6$. The desired conclusion follows.
\end{proof}

\begin{example}
Let $m=5$. Then the arc code $\C(\cA)$ over $\gf(32)$ has parameters $[33, 4, 30]$ and
weight enumerator
$$
1+ 169136z^{30} + 32736z^{31} + 508431z^{32} + 338272z^{33}.
$$
The subfield code  $\C(\cA)^{(2)}$ over $\gf(2)$ has parameters $[33, 11, 12]$ and is
distance-optimal.
The dual code $\C(\cA)^{(2)^\perp}$ over $\gf(2)$ has parameters $[33, 22, 5]$ and is
almost distance-optimal.
The extended code $\overline{\C(\cA)^{(2)^\perp}}$ over $\gf(2)$ has parameters $[34, 22, 6]$
and is distance-optimal.
\end{example}

\subsection{A comparison with the extended double error correcting codes}

The double error-correcting binary BCH code has parameters $[2^m-1, 2^m-1-2m, 5]$.
Its extended code has parameters $[2^m, 2^m-1-2m, 6]$. The information rate of this code is
$$
R_1':=\frac{2^m-2m-1}{2^m}.
$$
The information rate of the binary code $\overline{\C(\cA)^{(2)\perp}}$ is
$$
R_2':=\frac{2^{m}-2m}{2^m+2}.
$$
Thus the extended  double error-correcting binary BCH code has
almost the same information rate as that of the code $\overline{\C(\cA)^{(2)\perp}}$. Both codes
are distance optimal. In addition, the code  $\overline{\C(\cA)^{(2)\perp}}$
has new parameters. The reader is informed that a class of binary linear codes with
parameters $[2^m, 2^m-2m+1, 4]$ were reported in \cite{HD}.

\section{Concluding remarks}\label{sect-conlusion}

Optimal linear codes are very rare and precious. Optimal binary linear codes are rarer and
more precious. An interesting and difficult problem of
coding theory is to construct optimal codes, as it is much harder to construct optimal codes
in smaller fields. A more interesting and difficult problem
of coding theory is to construct optimal codes with new parameters.

The contribution of this paper is the two classes of distance-optimal binary linear codes
with respect to the sphere-packing bound. These codes are interesting, as their parameters
look new and they are distance-optimal. The two classes of arc codes employed in this paper are very special.
The Denniston arc codes over $\gf(2^m)$ are two-weight codes holding $2$-designs.
The maximal arc codes with parameters $[2^m+1, 4, 2m-2]$ are MDS codes over $\gf(2^m)$.
These codes were carefully selected, so that optimal binary linear codes have been obtained.
To obtain the optimal binary codes, we employed a combination of
coding techniques such as the subfield technique, the extension technique, and the
augmentation technique.

If two linear codes $\C_1$ and $\C_2$ over $\gf(q^m)$ are monomially equivalent, their
subfield codes $\C_1^{(q)}$ and $\C_2^{(q)}$ over $\gf(q)$ may not be equivalent. For
example, the class of maximal arc codes over $\gf(2^m)$ in \cite{WDT} are equivalent to
the class of Denniston arc codes, but their subfield codes over $\gf(2)$ are not
equivalent. To obtain linear codes over $\gf(q)$ with good parameters via the subfield
code technique,
one should select an extension filed $\gf(q^m)$ and a code over $\gf(q^m)$ carefully.

\end{document}